\documentclass[conference]{IEEEtran}
\usepackage{moreverb}
\usepackage{epsfig}
\usepackage{amsmath,amssymb,amsthm,mathrsfs,amsfonts,dsfont}
\usepackage{adjustbox,lipsum}
\usepackage{algorithm,algorithmic}
\usepackage{amsfonts}
\usepackage{epsfig}
\usepackage{calrsfs}
\usepackage{amssymb}
\usepackage{amsmath}
\usepackage{amsthm}
\usepackage{multirow}
\usepackage{setspace}
\usepackage{rotating}
\usepackage{graphicx}
\usepackage{tabularx}
\usepackage{array}
\usepackage{anyfontsize}
\usepackage{color,soul}
\usepackage{graphicx,dblfloatfix}
\usepackage{epstopdf}
\usepackage{blindtext}
\usepackage{amsmath}
\usepackage{amsthm,amssymb,amsmath,bm}
\usepackage{subfigure}
\usepackage{flushend}
\usepackage{subcaption}
\usepackage{amsfonts}
\usepackage{epsfig}
\usepackage{amssymb}
\usepackage{amsmath}
\usepackage{cite}
\usepackage{graphicx}
\usepackage{fancyhdr}
\usepackage[font={small}]{caption}
\usepackage{tabularx}
\usepackage{tcolorbox}
\usepackage{cite}
\usepackage{setspace}
\usepackage{soul}
\usepackage[a4paper,
            bindingoffset=0.2in,
            left=0.5in,
            right=0.5in,
            top=1.15in,
            bottom=1.1in,
            footskip=.2in
            ]{geometry}

\newtheorem{theorem}{Theorem}
\newtheorem{lemm}{Lemma}

\newtheorem{corollary}{Corollary}

\newtheorem{rem}{Remark}
\newtheorem{Pro}{Proposition}

\allowdisplaybreaks
\newcommand{\E}[1]{\mathbb{E}[#1]}

\newcommand{\nn}{\nonumber\\}
\newcommand{\Dbar}{\bar{D}}
\begin{document}

\title{AoI in M/G/1/1 Queues with Probabilistic Preemption
}
 \author{
 \IEEEauthorblockN{Mohammad~Moltafet\textsuperscript{*}, Hamid R. Sadjadpour\textsuperscript{*}, Zouheir~Rezki\textsuperscript{*}, Marian~Codreanu\textsuperscript{†}, and Roy~D.~Yates\textsuperscript{††}
 \\
\textsuperscript{*}Department of ECE, University of California Santa Cruz, USA 
(\{mmoltafe, hamid, zrezki\}@ucsc.edu)
\\
\textsuperscript{†}Department of Science and Technology, Link\"{o}ping University, Sweden (marian.codreanu@liu.se)
\\
\textsuperscript{††}Department of ECE, Rutgers University, USA (ryates@winlab.rutgers.edu)
}
 }

\maketitle
\begin{abstract}
We consider a status update system consisting of one source, one server, and one sink. The source generates packets according to a Poisson process and the packets are served according to a generally distributed service time. We consider a system with a capacity of one packet, i.e., there is no waiting buffer in the system, and model it as an M/G/1/1 queueing system. We introduce a probabilistically preemptive packet management policy and calculate the moment generating functions (MGFs) of the age of information (AoI) and peak AoI (PAoI) under the policy. According to the probabilistically preemptive policy, when a packet arrives, the possible packet in the system is replaced by the arriving packet with a fixed probability. Numerical results show the effectiveness of the packet management policy.

\end{abstract}	
\section{Introduction}\label{Introduction}
Ensuring the freshness of status updates for various real-world physical processes is a crucial component of many cyber-physical system applications, such as smart factories and autonomous vehicles. In these applications, stale status updates hold no value, even if they are successfully delivered to the destination or decision maker. In \cite{6195689}, the age of information (AoI)  
 was introduced to measure the information freshness at the destination in status update systems. The peak AoI (PAoI) was later introduced in \cite{6875100} as an alternative metric for evaluating information freshness. 
The AoI measures the difference between the current time and the time stamp of the last received sample from the monitored process. 

\subsection{Contributions}
We consider a single-source status update system with one server and no waiting buffer. Packets are generated according to a Poisson process and have  
with iid  general  service time distribution. We derive the moment generating functions (MGFs) of the AoI and PAoI under a probabilistically preemptive policy. According to the policy, when a packet arrives, the possible packet in the system is preempted by the arriving packet with a fixed probability. From the MGFs, we derive the average AoI and PAoI expressions. 

Our results show that using either a preemptive policy (an arriving packet always preempts a packet in service) or a non-preemptive policy (an arriving packet is blocked and cleared when the server is busy) can be quite inefficient in terms of AoI performance. When a preemption mechanism is available, the system must learn when to apply it. Moreover, finding the optimal value for the probability of preemption can significantly enhance timeliness.

\subsection{Related work}
The first queueing-theoretic study on AoI was presented in \cite{6195689}, where the average AoI for M/M/1, D/M/1, and M/D/1 first-come first-served (FCFS) queueing models was derived. 
As demonstrated in \cite{6310931,7415972}, implementing appropriate packet management policies in status update systems—either in the waiting queue, the server, or both—can significantly enhance information freshness.
The performance of various packet management policies in queueing systems with exponentially distributed service times and Poisson arrivals has been extensively analyzed in the literature \cite{8469047,8437591,8406966,8437907,9013935,9048914,9252168,9162681,Moltafet2020mgf,9611498,9705518}.

Besides exponentially distributed service time and Poisson arrivals, AoI has been studied under various arrival processes and service time distributions. 
 The distribution of the AoI and PAoI for the single-source PH/PH/1/1 and M/PH/1/2 queueing models were derived in \cite{9119460}.
The average AoI of a single-source D/G/1 FCFS  queueing model was derived in \cite{8406909}. A closed-form expression of the average AoI for a single-source  M/G/1/1 preemptive queueing model with hybrid automatic repeat request was derived in \cite{8006504}. The distributions of the AoI and PAoI of single-source  M/G/1/1 FCFS and G/M/1/1 FCFS queueing models were derived in \cite{8006592}.
A general formula for the distribution of the AoI in single-source single-server queueing systems was derived in \cite{8820073}. 
The average AoI and PAoI of a single-source status update system with Poisson arrivals and a service time with gamma distribution under the last-come first-served (LCFS) policy were analyzed in \cite{7541764}. 
The average AoI of a single-source G/G/1/1 queueing model was studied in \cite{9048933}. 
The distribution of the AoI in a generate-at-will single-source dual-server system was derived in \cite{Arxakar2024age}. The authors considered that the servers serve packets according to exponentially distributed service times where the sampling and transmission process is frozen for a period of time with  Erlang distribution upon each transmission.

The average AoI and PAoI in a multi-source M/G/1 FCFS queueing model were studied in    
\cite{9099557,inoue2024exact}.
The average AoI of a queueing system with two classes of Poisson arrivals with different priorities under a general service time distribution was studied in \cite{8886357}. 
The average AoI and PAoI of a multi-source M/G/1/1 queueing model under the globally preemptive packet management policy were derived in \cite{8406928}.  According to the globally preemptive packet management policy, a new arriving packet preempts the possible under-service packet independent of the source index of the packets.
The average AoI and PAoI of a multi-source M/G/1/1 queueing model under the non-preemptive policy were derived in \cite{9500775}. According to the non-preemptive policy, when the server is busy, any arriving packet, independent of their source index, is blocked and cleared.  
 The authors of \cite{9519697} considered a multi-source system with Poisson arrivals where the server serves packets according to a phase-type distribution. 
 They numerically obtain the distributions of the AoI and PAoI under a probabilistically preemptive policy. 
The MGFs of the AoI and PAoI of a multi-source M/G/1/1 queueing model under the self-preemptive policy were derived in \cite{9869867}. According to the self-preemptive policy, a new arriving packet preempts a possible under-service packet only if they have the same source index. In addition, the authors of \cite{9869867} derived the MGFs of the AoI and PAoI of the models studied in \cite{9500775} and \cite{9869867}.
The distributions of the AoI and PAoI for a generate-at-will multi-source system with a phase-type service time distribution were derived in \cite{10139823}.
The Laplace-Stieltjes transform of the AoI for a two-source system with Poisson arrivals and a generally distributed service time was derived in \cite{10038591}. The authors assumed that there is a buffer for each source and studied three versions of the self-preemptive policy.


\section{System Model and Main Results}\label{System Model and Summary of Results}
We consider a status update system consisting of one source, one server, and one sink.
The source generates status update packets about a random process and sends them to the sink.  Each status update packet contains the measured value of the monitored process and a time stamp representing the time when the sample was generated. We assume that packets are generated by a rate $\lambda$ Poisson process and have 
iid service times with PDF $f_U(\cdot)$ and moment generating function $M_{U} (s) = \mathbb{E}[e^{sU}]$. 
Finally, we restrict our attention to systems with capacity one (i.e., there is no waiting buffer); thus, the considered setup is a single-source M/G/1/1 queue. Next, we explain the packet management policy. 

\smallskip

\noindent\textbf{Probabilistically Preemptive Policy:} 
When the server is idle, an arriving packet
immediately goes into service. When the server is busy, an arriving packet preempts the packet in service 
with the probability $\theta$; otherwise, the arriving packet is discarded.

\subsection{AoI and PAoI Definition}
The AoI at the sink is defined as the time elapsed since the last successfully received packet was generated. Formally, let $t_{i}$ denote the time instant at which the $i$th delivered status update packet was generated, and let $t'_{i}$ denote the time instant at which this packet arrives at the sink. 
Hereafter, we refer to the $i$th delivered packet simply as ``packet $i$.''
Thus, at time $\tau$, the index of the most recently received packet is given by
$N(\tau)=\max\{i'\mid t'_{i'}\le \tau\},$
and the time stamp of the most recently received packet is $\xi(\tau)=t_{N(\tau)}.$
The AoI at the sink is defined as the random process
${\delta(t)=t-\xi(t).}$

Let the random variable
$Y_{i}=t'_{i}-t'_{i-1}$
denote the $i$th interdeparture time, i.e., the time elapsed between the departures of (delivered) packets  $i-1$ and $i$.  
Moreover, let the random variable
$T_{i}=t'_{i}-t_{i}$
denote the system time of packet $i$, i.e., the time duration that $i$th delivered packet spends in the system. 

 Another commonly used freshness metric is the PAoI \cite{6875100}, defined as the value of the AoI immediately before receiving an update packet. Accordingly, the PAoI associated with packet $i$, denoted by $A_{i}$, 
is given as
%
$A_{i}= Y_{i}+T_{i-1}.$
%

We assume that the considered status update system is stationary and the AoI process
is ergodic. Thus, we have $T_{i}=^{\mathrm{st}}T$,  $Y_{i}=^{\mathrm{st}}Y$, and $A_{i}=^{\mathrm{st}}A,\forall i$, where ${=^{\mathrm{st}}}$ means stochastically identical (i.e., they have an identical marginal distribution).
We present the main results of our paper in the following theorem.

\begin{theorem}\label{T_source-aware}
 The MGFs of the AoI and PAoI for the M/G/1/1 queueing model under the probabilistically preemptive packet management policy, denoted by $M_{\delta}(s)$ and $M_{A}(s)$, respectively, are given as
\begin{align}
&{M_{\delta}(s)}=\dfrac{M_U(s-\theta\lambda)(M_{Y}(s)-1)}{sM_{U}(-\theta\lambda)\bar{Y}},\\&
M_{A}(s)=\dfrac{M_U(s-\theta\lambda)M_{Y}(s)}{M_{U}(-\theta\lambda)},
\end{align}
where $M_{Y}(s)$ is the MGF of the interdeparture time $Y$, which is given as
\begin{align}
M_{Y}(s)=\dfrac{\lambda(\theta\lambda-s)M_U(s-\theta\lambda)}{(\lambda-s)(\theta\lambda M_U(s-\theta\lambda)-s)},
\end{align}  
 and $\bar{Y}$ is mean interdeparture time.
 %
\end{theorem}

\begin{rem}
The MGF of the (peak) AoI under the probabilistically preemptive policy, presented in Theorem~\ref{T_source-aware}, generalizes the existing results in \cite{9869867} for a single-source system. Specifically, by letting $\theta\rightarrow 0$, the MGF of the (peak) AoI becomes equal to the MGF of (peak) AoI for a single-source system under the non-preemptive policy derived in \cite{9869867}.   Moreover, by letting $\theta\rightarrow 1$, the MGF of the (peak) AoI becomes equal to the MGF of (peak) AoI for a single-source system under the preemptive policy derived in \cite{9869867}.
\end{rem}

\begin{rem}\label{rem1MGFage}
The $m$th moment of the (peak) AoI is derived by evaluating the $m$th derivative of the MGF of the (peak) AoI at $s=0$. 
\end{rem}

\begin{corollary}\label{agemg11theoremblock}
The average AoI and  PAoI of the M/G/1/1 queueing model under the probabilistically preemptive packet management  policy, denoted by $\Delta$ and $\bar{A}$, respectively, are given as
\begin{align}
&\Delta\!=\!\dfrac{\!M_{U}(\!-\theta\lambda)((\theta^2\!-\!\theta)(M_{U}(-\theta\lambda)\!+\!\!\lambda M'_{U} (-\lambda\theta))\!+\!\theta\!-\!1)\!+\!1}{\lambda M_{U}(-\theta\lambda)^2(\theta^2-\theta)+\lambda M_{U}(-\theta\lambda) \theta},
\\&
\bar{A}=\dfrac{M_{U}(-\theta\lambda)(\theta-1)+\lambda\theta M'_{U} (-\lambda\theta)+1}{\theta\lambda M_{U}(-\theta\lambda)},
\end{align}
where $M'_{U} (-\lambda\theta)=\mathbb{E}[Ue^{-\lambda\theta U}]$. 
%
\end{corollary}

\begin{rem}
The average AoI under the probabilistically preemptive policy, presented in Corollary~\ref{agemg11theoremblock}, generalizes the existing results in \cite{8006504}. Specifically, by letting $\theta\rightarrow 0$, the average AoI becomes equal to the average AoI for an M/G/1/1 system under the non-preemptive policy derived in \cite[Theorem~1]{8006504}. Moreover, by letting $\theta\rightarrow 1$, the average AoI becomes equal to the average AoI for an M/G/1/1 system under the preemptive policy derived in \cite[Theorem~5]{8006504}.
\end{rem}

\section{Derivation of the MGF of the (Peak) AoI }\label{Calculation of the MGF of the AoI under the packet management policies}
To prove Theorem \ref{T_source-aware}, we first provide Lemma~\ref{lemmsmgfage} which describes the MGFs of the  AoI and PAoI.
\begin{lemm}\label{lemmsmgfage}
The MGFs of the AoI $\delta$ and PAoI $A$ in an M/G/1/1 queue under the probabilistically preemptive policy 
can be characterized as
\begin{align}\label{MGFofagegeneral}
M_{\delta}(s)&=\dfrac{M_{T}(s)(M_{Y}(s)-1)}{s\bar{Y}},\\\label{MGFpeak}
M_{A}(s)&=M_{T}(s)M_{Y}(s), 
\end{align}
where $M_{T}(s)$ is the  MGF of the system time $T$.
\end{lemm}
\begin{proof} 
The proof follows similar steps as in \cite[Lemma~1]{9869867}.
\end{proof}

As shown in Lemma \ref{lemmsmgfage}, the main challenge in calculating the MGF  of the (peak) AoI reduces to deriving the  MGFs of the system time, $M_{T}(s)$, and interdeparture time, $M_{Y}(s)$. Note that when we have $M_{Y}(s)$, we can easily derive $\bar{Y}$, i.e., $
\bar{Y}=\mathrm{d}(M_{Y}(s))/{\mathrm{d}s}|_{s=0}
$.
Next, we calculate $M_{T}(s)$ and $M_{Y}(s)$ in Propositions \ref{Pro1} and \ref{Pro2}, respectively. 
\begin{Pro}\label{Pro1}
The MGF of the system time $T$
is 
\begin{align}\label{mgfsystemtime}
M_{T}(s)=\dfrac{M_U(s-\theta\lambda)}{M_{U}(-\theta\lambda)}.
 \end{align}
\end{Pro}
\begin{proof}
Let $D$ denote the event that a packet that goes into service is delivered. According to the packet management policy, the system time $T$ of a \textit{delivered packet} is equal to its service time. Thus, the distribution of $T$ satisfies
\begin{align}
&\mathrm{Pr}(t\le T<t+\epsilon)=\mathrm{Pr}(t\le U\le t+\epsilon\mid D)\nn
&\qquad=\frac{\mathrm{Pr}(t<U<t+\epsilon)\mathrm{Pr}(D\mid t<U<t+\epsilon)}{\mathrm{Pr}(D)}.
\end{align}
 Hence, $f_{T}(t)=\lim_{\epsilon\to0}\mathrm{Pr}(t\le T\le t+\epsilon)/\epsilon$ is calculated as 
 \begin{align}\label{Sys_time_G1}
  f_T(t)&= \dfrac{f_U(t)\mathrm{Pr}(D\mid U=t)}{\mathrm{Pr}(D)}.  
 \end{align}
%
%
Because packets arriving at a busy server are admitted into service as Bernoulli trials with probability $\theta$, the arrivals of preempting packets, as  seen by the busy server, are a thinned Poisson process with arrival rate $\theta\lambda$. The event  $D$ occurs if and only if this thinned arrival process has zero arrivals during the service period $U=t$; i.e., \begin{align}\label{con_pro_eve_eq}
\mathrm{Pr}(D\mid U=t)=e^{-\theta\lambda t}.
    \end{align}
This implies 
\begin{align}
\mathrm{Pr}(D)&=\int_{0}^{\infty}\mathrm{Pr}(D\mid U=t)f_U(t)\,dt\nn
&=\int_{0}^{\infty}e^{-\theta\lambda t}f_U(t)\,dt
=M_{U}(-\theta\lambda).
\label{con_pro_eve_eq1}
\end{align}
By substituting \eqref{con_pro_eve_eq} and \eqref{con_pro_eve_eq1} into \eqref{Sys_time_G1}, 
\begin{align}
    f_{T}(t)=\dfrac{f_U(t)e^{-\theta\lambda t}}{M_{U}(-\theta\lambda)},
    \label{f_t1_0}
\end{align}
and the claim follows since $M_{T}(s)=\mathbb{E}[e^{sT}]$. 
\end{proof}

The next step is to derive the MGF of the interdeparture time $Y$, $M_{Y}(s)$, which is carried out in the following.
\begin{Pro}\label{Pro2}
The MGF of the  interdeparture time $Y$
is 
\begin{align}\label{mgfinterde1}
M_{Y}(s)=\dfrac{\lambda(\theta\lambda-s)M_U(s-\theta\lambda)}{(\lambda-s)(\theta\lambda M_U(s-\theta\lambda)-s)}.
\end{align}
\end{Pro}
\begin{proof}
To calculate the  MGF 
${{M}_{Y}(s)=\mathbb{E}[e^{sY}]}$, we first  characterize $Y$ using a semi-Markov chain. The semi-Markov chain, shown in Fig.~\ref{Semi-Chain_c}, represents the dynamics of the system occupancy states and transition probabilities of different states concerning the interdeparture time $Y$.
When a packet is successfully delivered to the sink, the system goes to state $q_0$, waiting for a fresh packet. State $q_{1}$ indicates that a packet is under service. From the graph, the interdeparture time $Y$ is calculated by characterizing the required time to start from state $q_0$ and return to $q_0$. 
The transitions between the states of the graph in Fig.~\ref{Semi-Chain_c} are explained as follows.

\begin{enumerate}
\item $ q_0\rightarrow q_{1}$: The system is in the state $q_0$ and a packet arrives. This transition happens with probability $1$. We denote the sojourn time of the system in state $q_0$ before this transition by $\tilde\eta$, which is exponentially distributed with parameter $\lambda$, i.e., $f_{\tilde\eta}(t)=\lambda e^{-\lambda t}$.

\item $q_1\rightarrow q_0$: The system is in state $q_1$, i.e., serving a packet, and the packet completes service and is delivered to the sink. This transition happens with probability ${p}=\mathrm{Pr}(D)=M_{U}(-\theta\lambda)$ (see \eqref{con_pro_eve_eq1}). We denote the sojourn time of the system in state $q_1$ before this transition by ${\eta}$, which has the distribution ${\mathrm{Pr}({\eta}>t)=\mathrm{Pr}(U>t\mid D)}$.

\item $q_{1}\rightarrow q_{1}$: The system is in state $q_{1}$ and a fresh packet, possibly after several packets that were blocked and cleared, arrives and preempts the under-service packet. 
Let $\Dbar$ denote the event where a packet under service is not delivered because it is preempted (and 
and thus the transition $q_{1}\rightarrow q_{1}$ occurs.) This transition happens with probability $\bar p=\mathrm{Pr}(\Dbar)$. Since $\Dbar$ is the complement of $D$, $\bar p=1-M_{U}(-\theta\lambda)$. We denote the sojourn time of the system in state $q_{1}$ before this transition by $\bar\eta$.
\end{enumerate}

\begin{figure}
\centering
\includegraphics[width=.4\linewidth,trim = 0mm 0mm 0mm 0mm,clip]{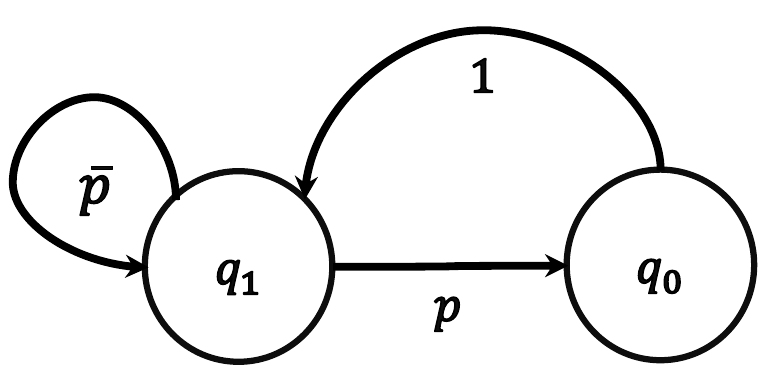}\vspace{-2mm}
\caption{The semi-Markov chain corresponding to the interdeparture time, $Y$. }  
\label{Semi-Chain_c}
\vspace{-4mm}
\end{figure}

\begin{figure*}[h]
\centering
\subfigure[$\lambda=0.2$]{
\includegraphics[width=0.22\textwidth]{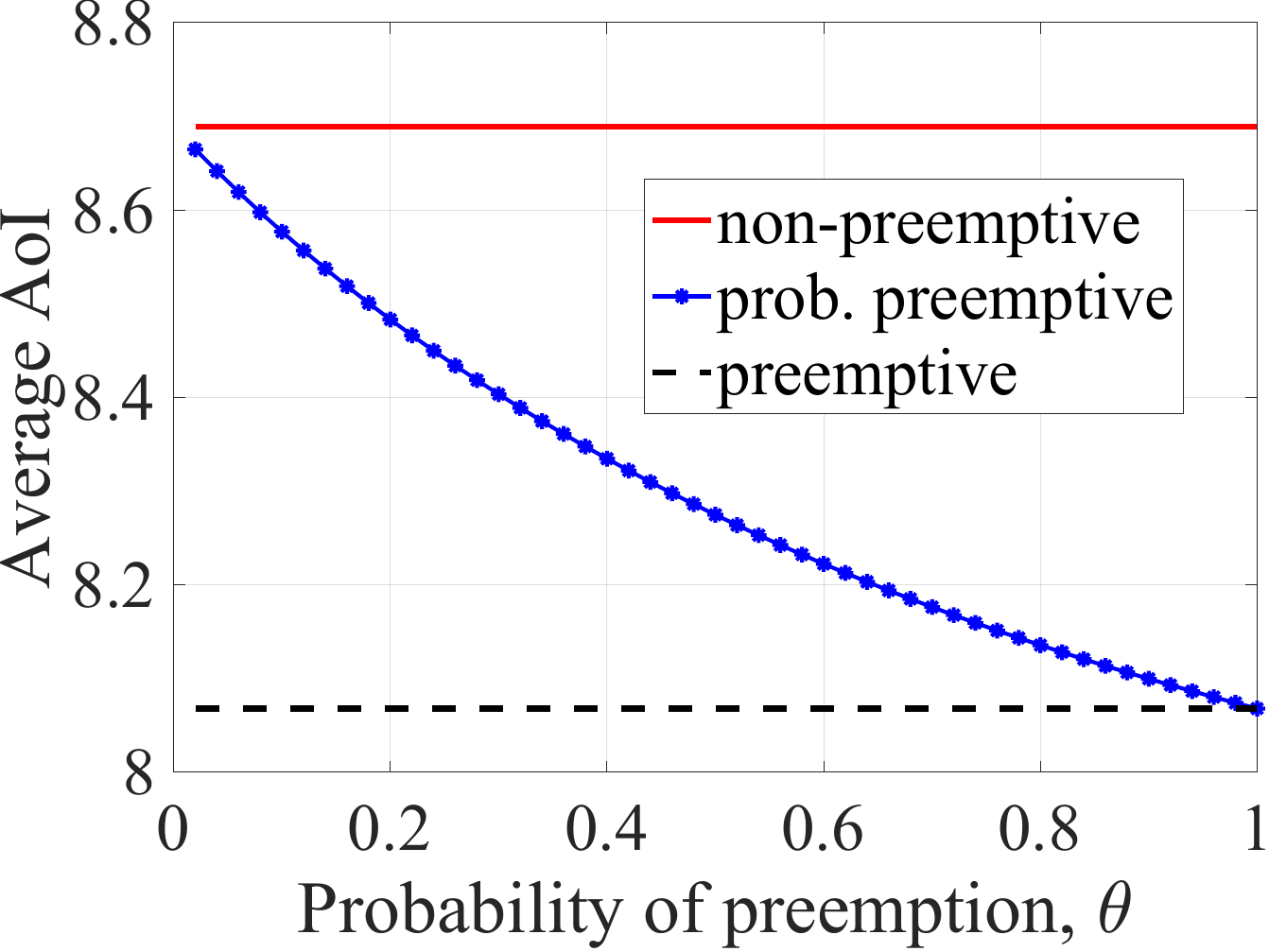}
\label{FP2}
}
\subfigure[$\lambda=1$]
{
\includegraphics[width=0.22\textwidth]{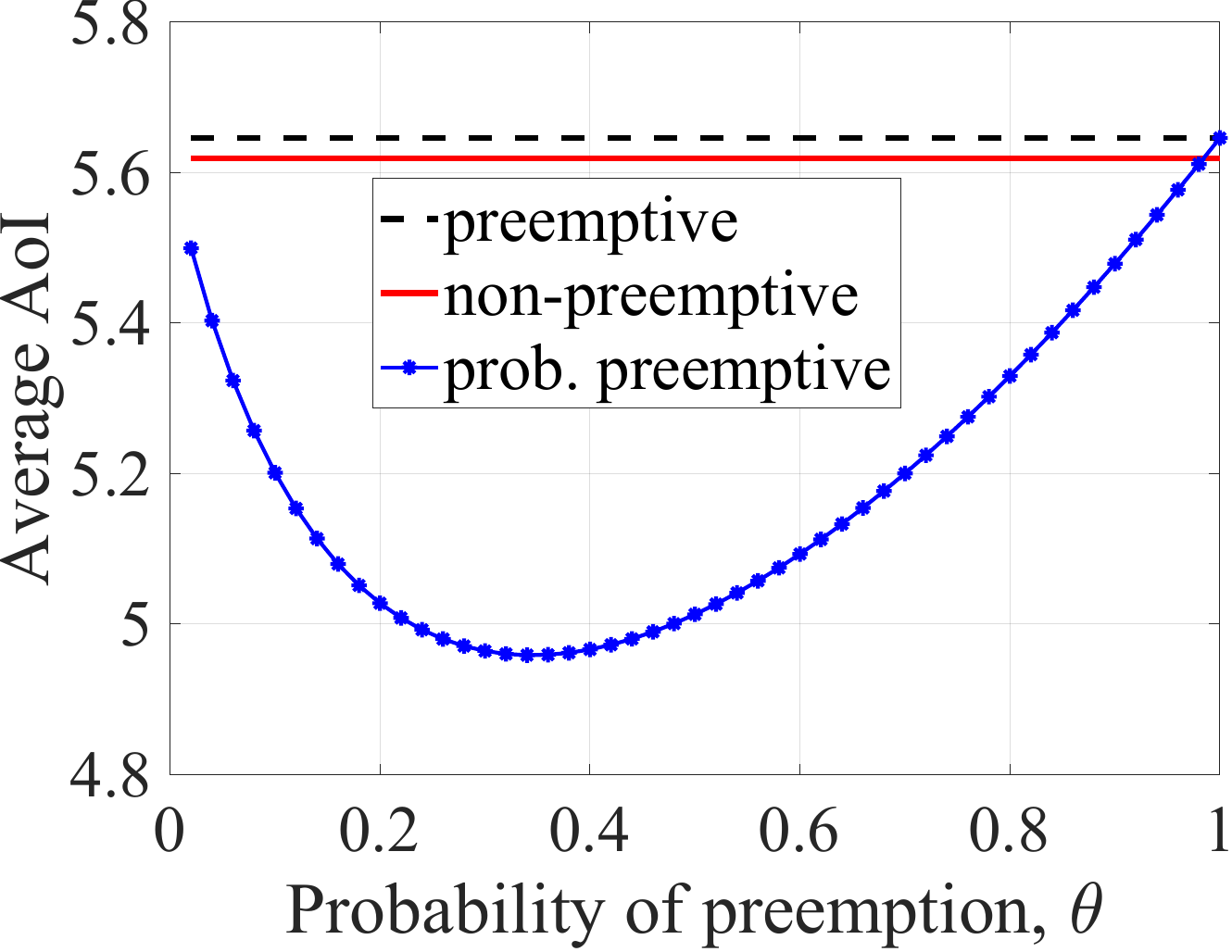}
\label{F1}
}
\subfigure[$\lambda=2$]{
\includegraphics[width=0.22\textwidth]{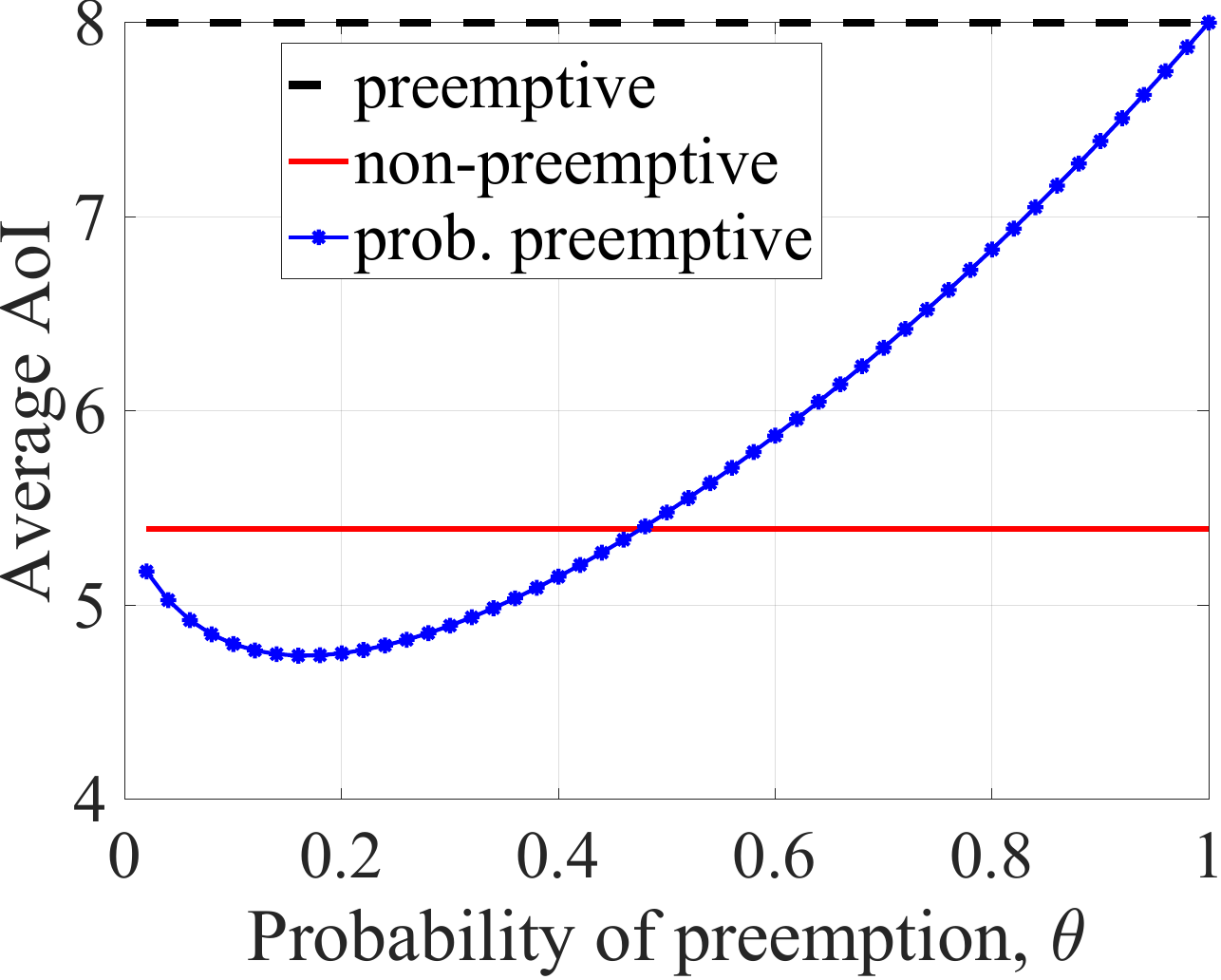}
\label{F2}
}
\subfigure[$\lambda=3$]{
\includegraphics[width=0.22\textwidth]{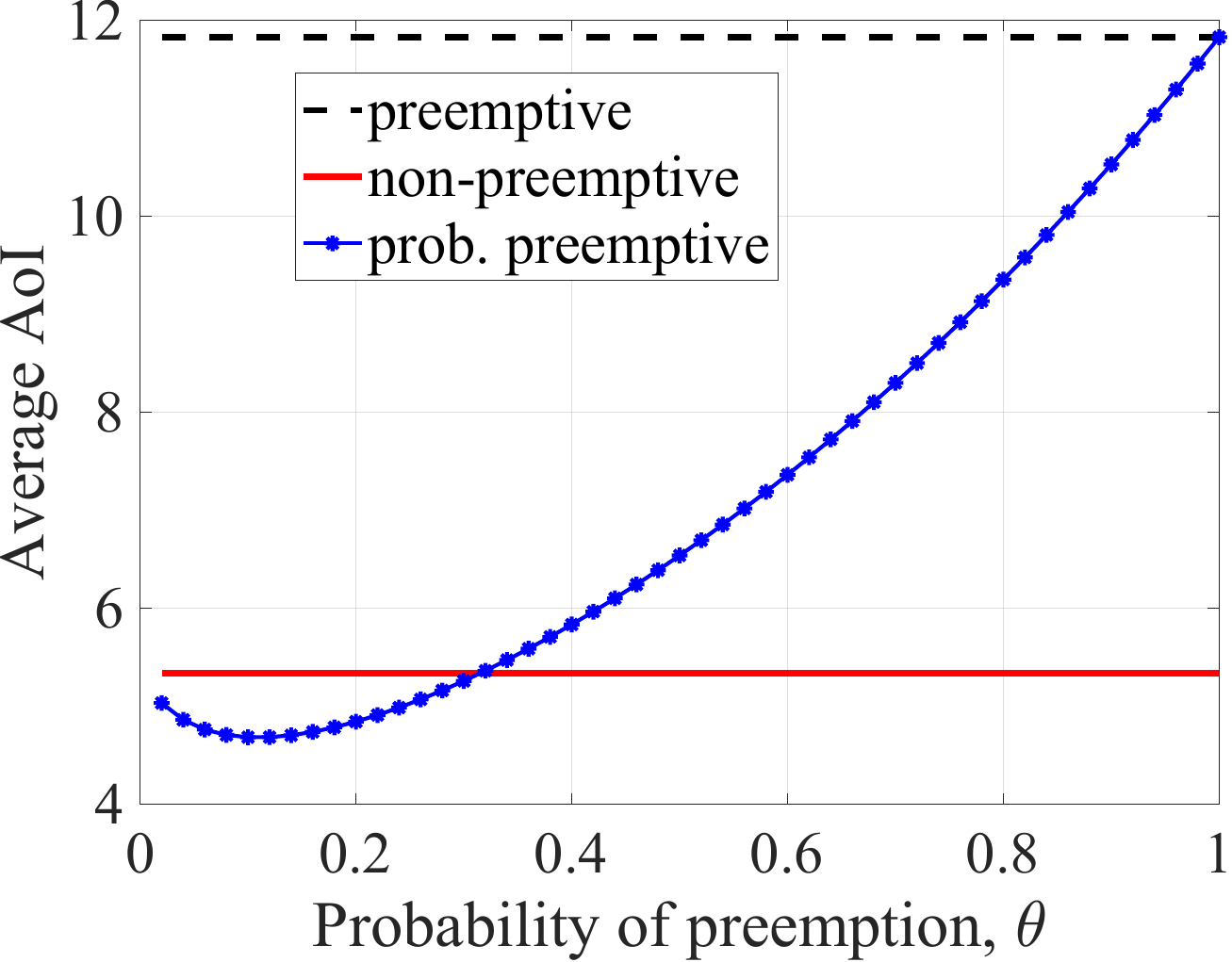}
\label{F3}
}
\caption{The average AoI of different policies as a function of the probability of preemption $\theta$. 
}
\label{Versus_P}
\vspace{-3mm}
\end{figure*}

\begin{figure*}[h]
\centering
\subfigure[$\lambda=0.2$]{
\includegraphics[width=0.22\textwidth]{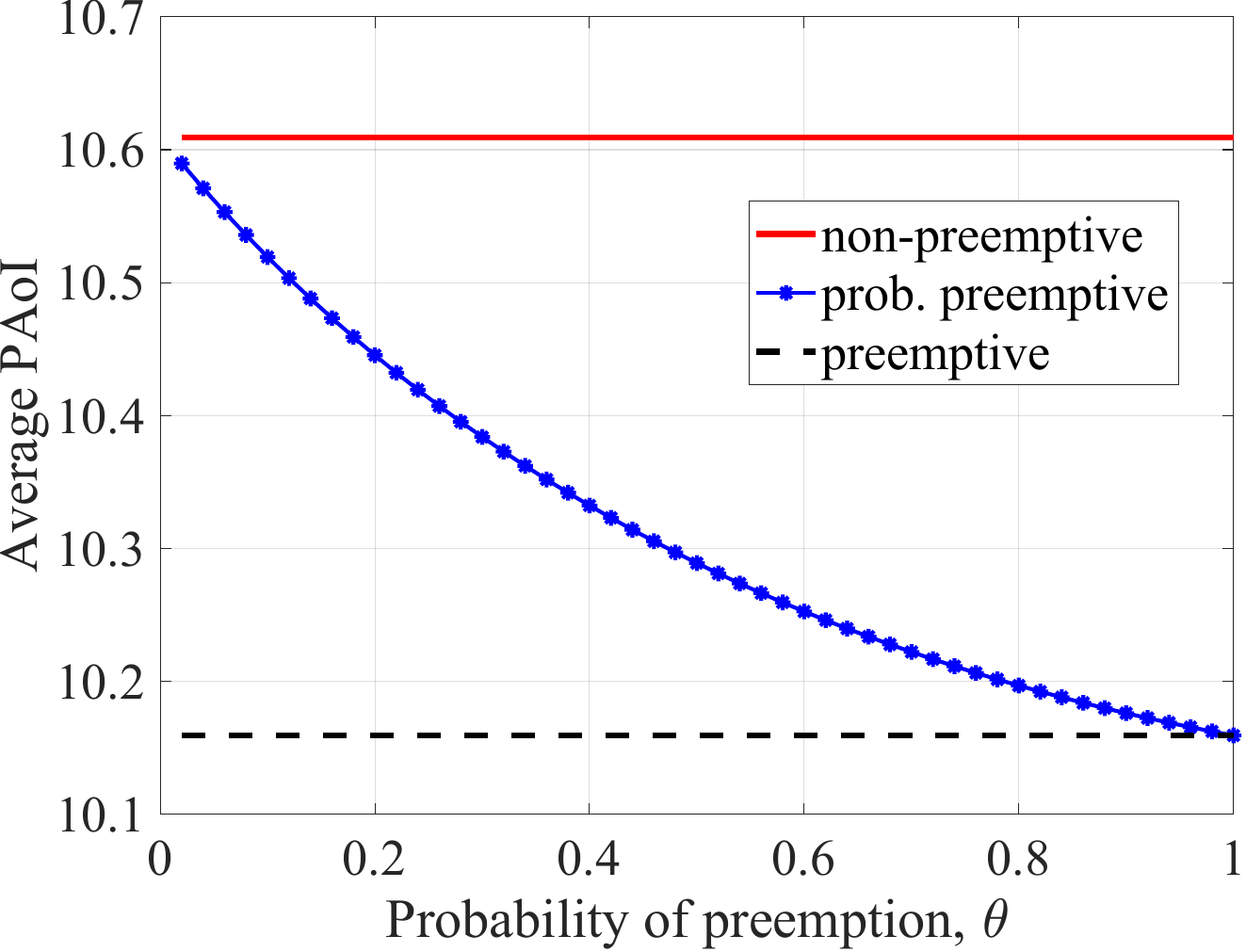}
\label{PFP2}
}
\subfigure[$\lambda=1$]
{
\includegraphics[width=0.22\textwidth]{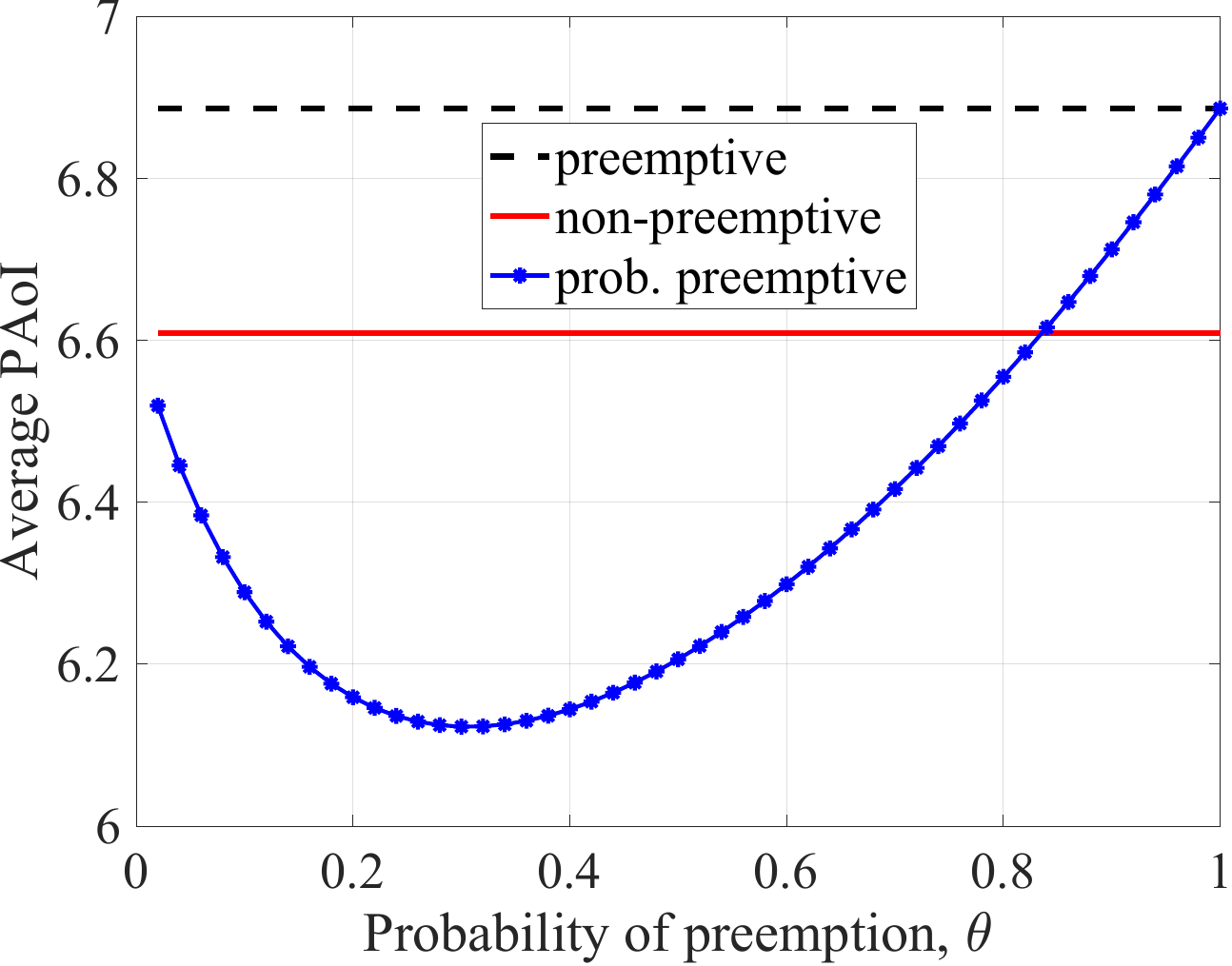}
\label{PF1}
}
\subfigure[$\lambda=2$]{
\includegraphics[width=0.22\textwidth]{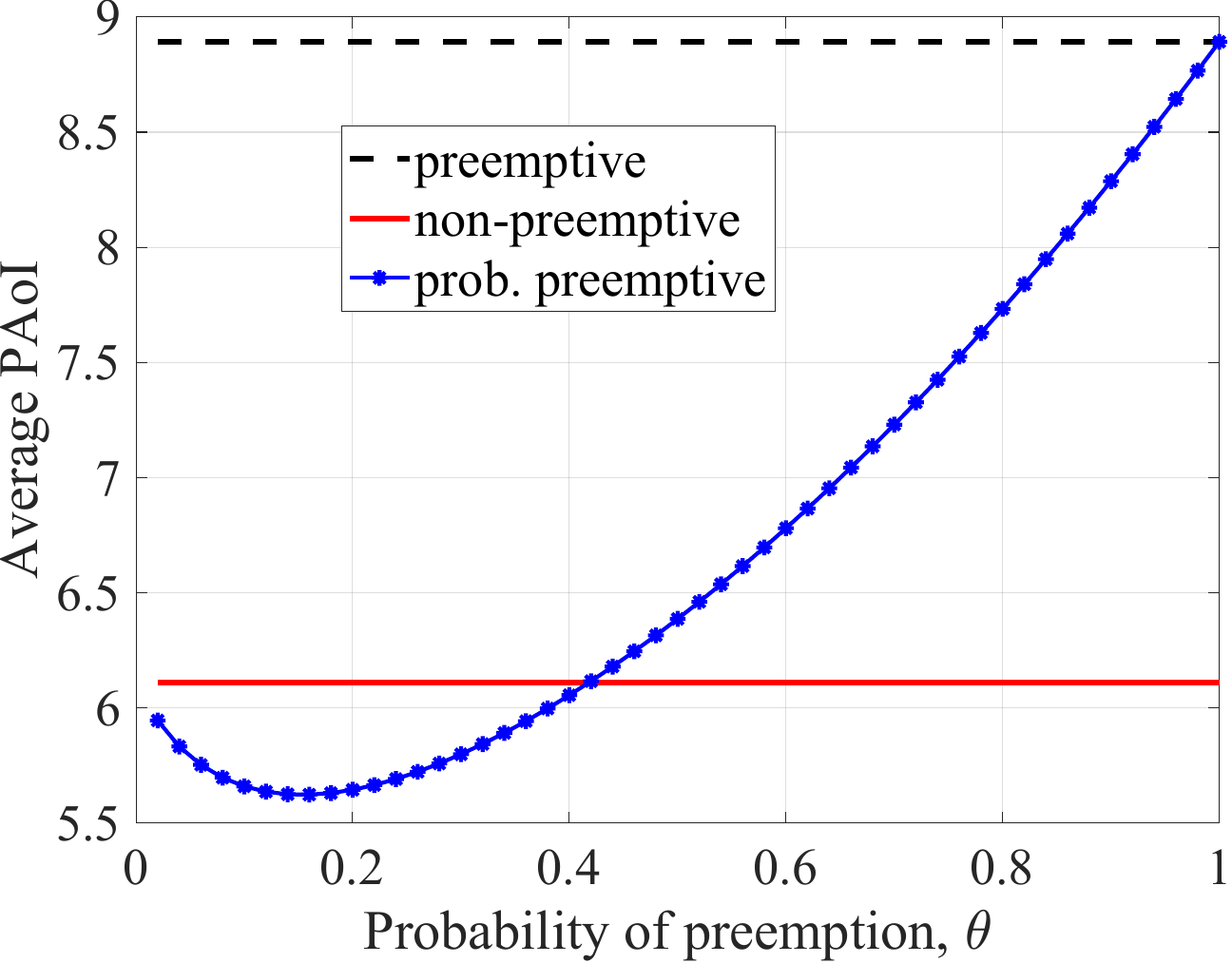}
\label{PF2}
}
\subfigure[$\lambda=3$]{
\includegraphics[width=0.22\textwidth]{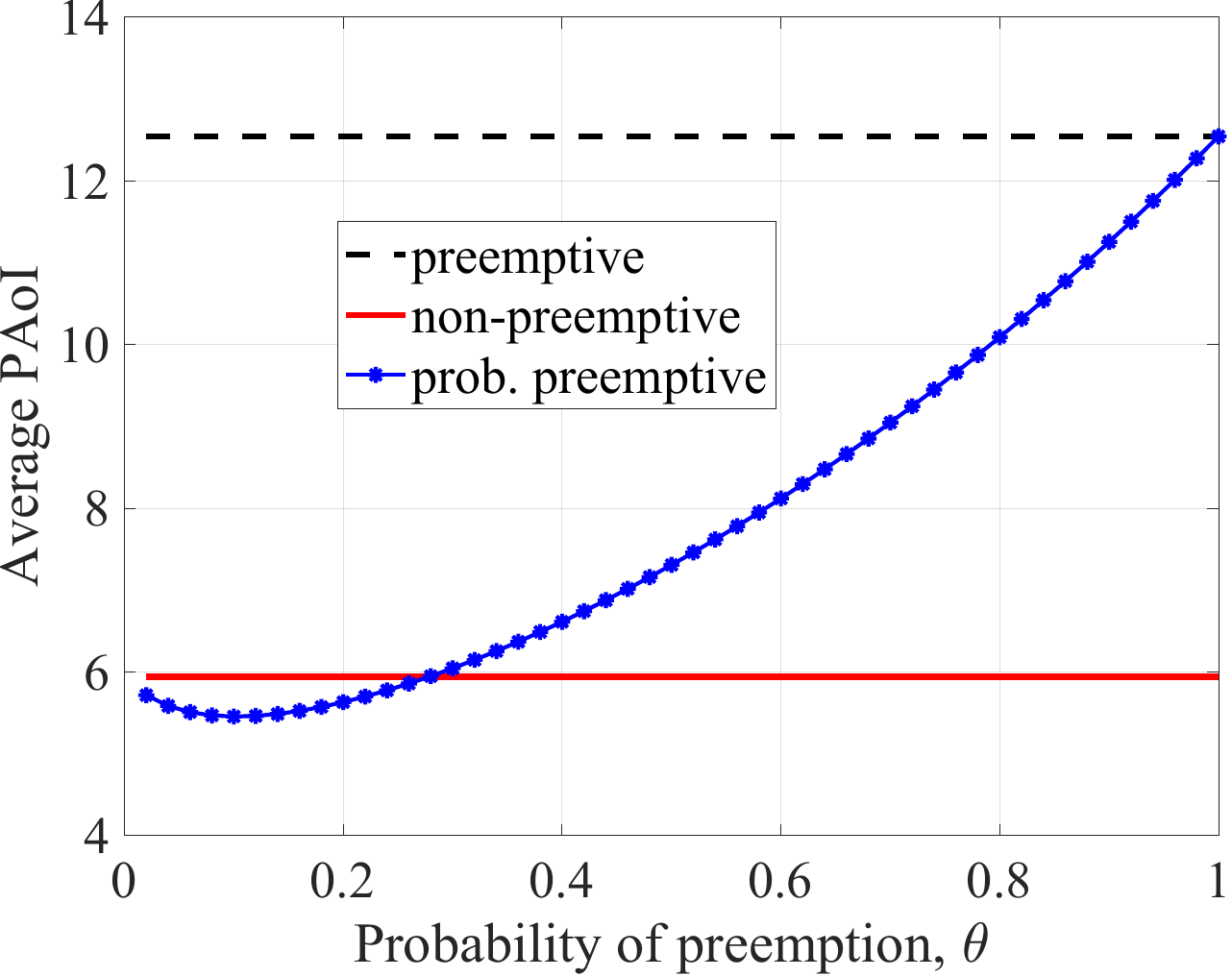}
\label{PF3}
}\vspace{-2mm}
\caption{The average PAoI of different policies as a function of the probability of preemption $\theta$. 
}
\label{Peak_Versus_P}
\vspace{-4mm}
\end{figure*}

In the following lemma, we calculate the sojourn time distributions. 
\begin{lemm}\label{lemmfaap}
The PDFs of the random variables $\eta$ and $\bar\eta$ are 
\begin{align}
f_{\eta}(t)=\dfrac{f_U(t)e^{-\theta\lambda t}}{M_{U}(-\theta\lambda)},
\quad
f_{\bar\eta}(t)=\dfrac{\theta\lambda e^{-\theta\lambda t}(1-F_U(t))}{1-M_{U}(-\theta\lambda)}.
\end{align}
\end{lemm}
\begin{proof}
The PDF of the random variable $\eta$ is the same as that of  the system time $T$, which was derived in \eqref{f_t1_0}.
To derive the PDF $\bar\eta$, we observe that packets arriving at the busy server are admitted into service as Bernoulli trials with probability $\theta$. Thus, the arrival time $R$ of a preempting packet, as  seen by the busy server, is the first arrival time of a thinned Poisson process with arrival rate $\theta\lambda$. Thus, $R$ has an exponential PDF with parameter $\theta\lambda$.
The event $\Dbar$ occurs when a preemption occurs prior to the completion of the under-service packet. It follows that
%
%
$\bar\eta$ has PDF 
\begin{align}\nonumber
f_{\bar\eta}(t)
&\stackrel{}{=}\lim_{\epsilon\rightarrow 0}\dfrac{\mathrm{Pr}(t<R<t+\epsilon\mid \Dbar)}{\epsilon}\\&\nonumber
=\lim_{\epsilon\rightarrow 0}\dfrac{\mathrm{Pr}(t<R<t+\epsilon)\mathrm{Pr}(\Dbar\mid t<R<t+\epsilon)}{\epsilon\mathrm{Pr}(\Dbar)}\nn
&=\dfrac{\mathrm{Pr}(U>t)}{\mathrm{Pr}(\Dbar)}\lim_{\epsilon\rightarrow 0}\dfrac{\mathrm{Pr}(t<R<t+\epsilon)}{\epsilon}\nn
&\label{f_t1}
=\dfrac{1-F_U(t)}{1-M_{U}(-\theta\lambda)}
f_R(t).
\end{align}
Finally, by substituting the exponential $(\theta\lambda)$ PDF of $R$, the PDF of $f_{\bar\eta}(t)$ is verified. 
%
\end{proof}
As shown in Fig.~\ref{Semi-Chain_c}, the interdeparture time between two consecutive delivered packets is equal to {the total sojourn time} experienced by the status update system between starting from state $q_0$ and returning to state $q_0$. 
The total sojourn time starting from $q_0$ and returning to $q_0$ consists of a summation of the individual sojourn times in each state for all possible paths $\{q_0,\ldots,q_0\}$. 
Thus, the random variable $Y$ can be characterized by random variables $\tilde\eta$, $\eta$, and $\bar\eta$, i.e., sojourn times in different states, and their numbers of occurrences.    
Let the random variable $V$
denote the numbers of occurrences of $\bar\eta$, between starting from $q_0$ and returning to $q_0$. Consequently, $Y$ can be presented as 
\begin{equation}\label{Y_c_sojourn}
Y=\tilde\eta+V\bar\eta+\eta.
\end{equation} 
As an example, a path starting from $q_0$ to $q_0$ could be $\{q_0\rightarrow q_1,q_1\rightarrow q_1,q_1\rightarrow q_1,q_1\rightarrow q_0\}$ which occurs with probability $p\bar{p}^2$.
Using $Y$ defined in \eqref{Y_c_sojourn}, 
 ${M}_{Y}(s)$ can be calculated as 
\begin{align}\nonumber
{M}_{Y}(s)=\mathbb{E}[e^{s(\tilde\eta+V\bar\eta+\eta)}]
&\stackrel{(a)}{=}\sum_{v=0}^{\infty}p\bar{p}^v\mathbb{E}[e^{s\tilde\eta}]\mathbb{E}[e^{s\bar\eta}]^v\mathbb{E}[e^{s\eta}]
\\&\label{mgfequ}
=\dfrac{{p}\mathbb{E}[e^{s\tilde\eta}]\mathbb{E}[e^{s\eta}]}{1-\bar{p}\mathbb{E}[e^{s\bar\eta}]},
\end{align}
where equality $ (a) $ follows because the sojourn times in different states, i.e., $\tilde\eta$, $ \eta$, and $\bar\eta$, and $V$ are all mutually independent.
Next, we present the 
values of 
$\mathbb{E}[e^{s\tilde\eta}],~\mathbb{E}[e^{s\eta}]$, and $\mathbb{E}[e^{s\bar\eta}]$.

\begin{lemm}\label{rem01}
By using the PDFs presented in  Lemma \ref{lemmfaap}, the values of $\mathbb{E}[e^{s\tilde\eta}],~\mathbb{E}[e^{s\eta}]$, and $\mathbb{E}[e^{s\bar\eta}]$ are given as
\begin{align}\nonumber
&\mathbb{E}[e^{s\tilde\eta}]=\dfrac{\lambda}{\lambda-s},~~~~~~~
\mathbb{E}[e^{s\eta}]=\dfrac{M_U(s-\theta\lambda)}{M_{U}(-\theta\lambda)},\\&\label{MGfprob}
\mathbb{E}[e^{s\bar\eta}]=\dfrac{\theta\lambda(1-M_U(s-\theta\lambda))}{(M_{U}(-\theta\lambda)-1)(s-\theta\lambda)}.
\end{align}
\end{lemm}

Substituting \eqref{MGfprob} into \eqref{mgfequ} completes
the proof.
\end{proof}

Finally, substituting the MGF of the system time derived in  \eqref{mgfsystemtime} and the MGF of the interdeparture time derived in \eqref{mgfinterde1}  into \eqref{MGFofagegeneral}  results in the MGF of the AoI, ${M}_{\delta}(s)$, given in Theorem~\ref{T_source-aware}. In addition, substituting   \eqref{mgfsystemtime} and \eqref{mgfinterde1} into \eqref{MGFpeak} results in the MGF of the PAoI, ${M}_{A}(s)$, given in Theorem~\ref{T_source-aware}.

\section{Numerical Results}\label{Numerical Results}
In this section, using Corollary~\ref{agemg11theoremblock}, we evaluate the average (peak) AoI under the probabilistically preemptive policy when the service time $U$ 
has the log-normal PDF 
$f_U(t)=\frac{1}{t\omega\sqrt{2\pi}}\exp\left(-\frac{(\ln{t}-\alpha)^2}{2\omega^2}\right)\!,~t>0,$
for parameters $ \alpha\in(-\infty,\infty),$ and ${\omega>0}$. The  expected service time is  $\E{U}=\exp(\alpha+{\omega^2}/{2})$.

 Fig.~\ref{Versus_P} (resp. Fig.~\ref{Peak_Versus_P}) depicts the average AoI (resp. PAoI) under different policies as a function of the probability of preemption $\theta$ for different values of $\lambda$. The service time parameters are set as  $\alpha=0.75$ and $\omega=0.75$. For $\lambda=0.2$, the results are shown in Fig.~\ref{FP2} (resp. Fig.~\ref{PFP2}); for $\lambda=1$, in Fig.~\ref{F1} (resp. Fig.~\ref{PF1}); for $\lambda=2$, in Fig.~\ref{F2} (resp. Fig.~\ref{PF2}); and for $\lambda=3$, in Fig.~\ref{F3} (resp. Fig.~\ref{PF3}).
 From the figures, when $\theta\rightarrow 1$ the probabilistically preemptive policy performs as the preemptive policy, and when $\theta\rightarrow 0$ it performs as the non-preemptive policy, as expected. As it can be seen, by choosing a proper value for the probability of preemption $\theta$ the average AoI can be minimized, e.g., for $\lambda=0.2$ the minimum average AoI is achieved by $\theta=1$, and for $\lambda=1$ the minimum average AoI is achieved by $\theta=0.34$. It is worth noting that a similar behavior is observed for different parameters of the considered service time distribution and for other service time distributions, e.g., gamma distribution. 




\vspace{-3mm}
\section{Conclusions}\label{Conclusions}
We considered a single-source M/G/1/1 queueing system and derived the MGFs of the AoI and PAoI under the probabilistically preemptive policy. Using the MGF of the (peak) AoI, we studied the average (peak) AoI under a log-normally distributed service time. The results showed that by using an appropriate value for the probability of preemption, the system's performance can be significantly improved.

\bibliographystyle{IEEEtran}
\bibliography{Bibliography}
\end{document}